\documentclass{article}
\usepackage{amssymb,amsthm, amsmath, mathtools}
\usepackage[hidelinks]{hyperref}

\usepackage{geometry}
\newtheorem{Th}{Theorem}[section]
\newtheorem{Lemme}[Th]{Lemma}
\newtheorem{Prop}[Th]{Proposition}

\newcommand{\F}{\mathcal{F}}
\newcommand{\R}{\mathbb{R}}

\newcommand{\Z}{\mathbb{Z}}
\newcommand{\N}{\mathbb{N}}
\newcommand{\pth}[1]{\left(#1\right)}
\DeclareMathOperator{\conv}{\mathsf{conv}}
\DeclareMathOperator{\rad}{r}
\DeclareMathOperator{\ch}{ch}
\DeclareMathOperator{\cl}{cl}
\DeclareMathOperator{\h}{h}
\DeclareMathOperator{\HC}{HC}

\title{A fractional Helly theorem for set systems with slowly growing homological shatter function}
\author{Marguerite Bin\footnote{Universit\'e de Lorraine, CNRS, INRIA, LORIA,  F-54000 Nancy, France. {\tt marguerite.bin@loria.fr}}}
\date{\today}

\begin{document}
\maketitle

\begin{abstract}
    We study parameters of the convexity spaces associated with families of sets in $\R^d$ where every intersection between $t$ sets of the family has its Betti numbers bounded from above by a function of $t$. Although the Radon number of such families may not be bounded, we show that these families satisfy a fractional Helly theorem. To achieve this, we introduce graded analogues of the Radon and Helly numbers. This generalizes previously known fractional Helly theorems.
    
\end{abstract}
\section{Introduction}
Intersection patterns of convex sets of $\mathbb{R}^d$ enjoy many remarkable properties: for example, Helly's theorem \cite{radon1921mengen},\cite[$\mathsection 1.3$]{matousek2013lectures} states that the intersection of all sets of a finite family is nonempty if and only if every $d+1$ members of the family intersect. 
These properties are not specific to convex sets:  
lattice convex sets \cite{Doignon1973ConvexityIC}, 
or even good covers \cite{helly1930systeme} satisfy similar properties. The notion of convexity spaces was introduced in order to study these properties in greater generality, notably by recasting the convex hull operator. \medskip

Formally, a convexity space on a ground set $X$ is a family $\mathcal{C}$ of subsets of $X$ containing $\emptyset$ and $X$, closed under intersections and nested unions of chains. This framework allows us to define the convex hull of a set $P\subset X$, denoted $\conv_{\F}(P)$, as the intersection between all sets of $\mathcal{C}$ containing $P$. (Since $\mathcal{C}$ is closed by intersection, it is the smallest set of $\mathcal{C}$ containing $P$.) This allows us to define three parameters in particular.

\begin{itemize}
    \item[--] The \textit{Radon number} $\rad(\mathcal{F})$ of a convexity space $\mathcal{F}$, denoted $\rad(\mathcal{F})$, is the smallest integer $r$ such that every set $S\subset X$ of cardinality $r$ can be split into two nonempty disjoint parts $S=P_1\sqcup P_2$ satisfying $\conv_{\mathcal{F}}(P_1) \cap \conv_{\mathcal{F}}(P_2) \neq \emptyset$. We say such a partition is an $\F$-Radon partition of the point set $P$.
    \item[--] The \textit{Helly number} of a convexity space $\mathcal{F}$, denoted $\h(\mathcal{F})$, is the smallest integer $h$ with the following property: if in a finite subfamily $S\subset\mathcal{F}$ every $h$ members of $S$ have a point in common, then all members of $S$ have a point in common. If no such $h$ exists, we put $\h(\mathcal{F})=\infty$.
    \item[--] The \textit{colorful Helly number} of a convexity space $\F$, denoted $\ch(\F)$, 
    is the minimal number of colors $m$ such that for every coloring of a subfamily $\F'\subset \F$ with $m$ colors, if every colorful subfamily (one of each color) has nonempty intersection, then there is a color such that all elements of this color have nonempty intersection.
\end{itemize}
Here we also introduce, as suggested by \cite{Holmsen2021}:
\begin{itemize}
    \item[--] The \textit{$k$-th clique number} of a convexity space $\F$, denoted $\cl_k(\F)$ is the smallest integer $s$ such that for every finite subfamily $\F'\subset \F$, whenever a constant fraction of the $s$-tuples of $\F'$ intersect, some constant fraction $\mathcal{G}$ of $\F'$ forms a clique in the sense that every $k$ elements of $\mathcal{G}$ intersect. 
\end{itemize}
Several relations hold between such parameters for every convexity space. For example, Levi's inequality \cite{Levi_radon_implies_helly} states that $\h(\mathcal{F})\leq \rad(\mathcal{F})-1$. (We mention that the opposite relation is untrue, as later described in the proof of Lemma~\ref{ex:growing-radon-number}, adapted from \cite[Example 1]{Kay1971AxiomaticCT}.) 
The clique number can be bounded from above by the colorful Helly number~\cite{Holmsen2020}, which in turn is bounded from above by a function of the Radon number~\cite{Holmsen2021}.
\medskip

When the ground set is a topological space, yet another parameter emerges: the homological complexity $\HC_h(\F)$, defined as the maximum among the $h$ first Betti numbers of the members of the convexity space. When the ground set is $\R^d$, Goaoc and al. \cite{hellyBetti} first proved that a function of the homological complexity bounds from above the Helly number. Patáková \cite{patakova2022bounding} later improved the result to prove that it also bounds from above the Radon number. Combined with the previous work of Holmsen and Lee \cite{Holmsen2021}, it induces that a finite homological complexity entails a finite clique number, which was later improved in \cite{steppingUp} and generalized to topological set systems with one forbidden homological minor.

Making a sidestep, we can introduce a new parameter called the ($h$th) \textit{homological shatter function}, as suggested (implicitly) by Kalai and Meshulam in \cite{Kalai2010CombinatorialEF}:
\[\phi_{\mathcal{F}}^{(h)}\colon\left\{\begin{array}{ccc}
        \mathbb{N} & \rightarrow & \mathbb{N} \cup \{\infty\} \\
         k& \mapsto & \sup \left\{\tilde{\beta_i}(\bigcap\limits_{F\in \mathcal{G}} F,\Z_2)\mid \mathcal{G}\subset \mathcal{F}, |\mathcal{G}|\leq k, 0\leq i\leq h\right\}
    \end{array}\right..\]
with $\tilde\beta_i(X,\Z_2)$ denoting the $i$-th reduced Betti number of the topological space $X$. 
It describes the topological complexity of intersections between increasingly more sets. It is kind of a graded analogue of the homological complexity, and it encourages us to define graded analogues of other parameters. Note that this parameter is constant equal to the homological complexity when $\mathcal{F}$ is a convexity space. For this parameter to be relevant, we therefore need to work with families that are not convexity spaces\footnote{One could also stratify the sets of a convexity space $\mathcal{F}$ and define graded parameters for each stratum, or work with families of convexity spaces (and consider multiple convexity spaces for each layer).}: thankfully, the definitions of $\mathcal{F}$-hull, Radon number and Helly number remain reasonable when we switch from $\mathcal{F}$ being a convexity space to $\mathcal{F}$ being a family of sets. We refer the reader to \cite[$\mathsection 5.3$]{patakova2022bounding} for an exposition of the nuances between convexity spaces and general set systems. \medskip

In this note we make the following contributions:
\begin{itemize}
    \item[--] We introduce graded analogues of the Radon number, Helly number, and colorful Helly number. We relate these graded parameters to each other and to some of their ungraded analogues. 
    \item[--] We relate these graded parameters to the homological shatter function to refine the results of \cite{patakova2022bounding, steppingUp} mentioned previously. When the homological complexity is finite, the homological shatter function is stationary. We show that even when the homological shatter function is non-stationary, we can still guarantee a finite clique number as long as we can bound from above one specific value. This is a step toward a conjecture of Kalai and Meshulam \cite[Conjectures~6~and~7]{Kalai2010CombinatorialEF}, which suggests that a family with a homological shatter function that grows polynomially has a finite clique number. 
\end{itemize}

Some ad-hoc examples of topological set systems where our Theorem~\ref{th:main-theorem} applies but \cite[Corollary~1.3]{steppingUp} does not are later discussed Section~\ref{section:examples-betti-borne}.

\section{Graded parameters}

The homological shatter function can equivalently be defined as 
\[\phi_\mathcal{F}^{(h)}(t)= \sup_{\substack{\mathcal{F}'\subset\mathcal{F}\\ |\mathcal{F}'|=t} }\HC_{h}(\mathcal{F}'^{\cap}) \qquad \text{where} \qquad \F'^{\cap} \coloneqq \left\{\bigcap\limits_{F\in\mathcal{G}} F \mid \mathcal{G}\subset \F'\right\}.\]
A natural way to adapt the statements proved in \cite{patakova2022bounding} and \cite{Holmsen2021} and to relate the homological shatter function to other parameters is to introduce graded analogues of other parameters. We can similarly define the $t$-th graded Radon, Helly and colorful Helly numbers of the family $\F$ as:
 \[\rad^{(t)}(\mathcal{F})\coloneqq \sup \limits_{\substack{\mathcal{F}'\subset \mathcal{F} \\ |\mathcal{F}'|\leq t}} \rad (\mathcal{F}'),
 \quad
 \h^{(t)}(\mathcal{F})\coloneqq \sup\limits_{\substack{\mathcal{F}'\subset \mathcal{F} \\ |\mathcal{F}'|\leq t}} \h(\mathcal{F}'), 
\quad
 \ch^{(t)}(\F)\coloneqq\sup\limits_{\substack{\mathcal{F}'\subset \mathcal{F} \\ |\mathcal{F}'|\leq t}} \ch(\mathcal{F}').\]
These definitions lead us to consider the Radon numbers for finite set systems. 
We point out that considering two elements of $X$ as equivalent when they belong exactly to the same sets $F\in\mathcal{F}$ does not change the Radon and Helly numbers. In particular, we get that  $\rad^{(t)}(\mathcal{F})< 2^t$ from the pigeonhole principle; we sharpen this below (see Proposition~\ref{prop:graded-radon-bounded}).  \medskip

These graded numbers follow the same relations as their ungraded analogues: we can consider all subfamilies of size $t$ and apply the known relations between the ungraded numbers. For example, we get the following relations for a family $\mathcal{F}$:
\begin{align}
    \h^{(t)}(\mathcal{F})\leq \rad^{(t)}(\mathcal{F})-1 && \text{ if } \bigcap\mathcal{F}=\emptyset 
    \label{eq:levi}\\
    \ch^{(t)}(\F) \leq m\pth{\rad^{(t)}(\F)} && \text{with }m(\cdot) \text{ defined by \cite[Lemma 2.3]{Holmsen2021}}
    \label{eq:colorful-helly}
\end{align}
by applying \cite{Levi_radon_implies_helly} for the first relation and \cite[Lemma 2.3]{Holmsen2021} for the second. \medskip

We present two bridges between graded and non-graded parameters.

\paragraph{Growth rate of the graded numbers.} The first bridge is straightforward and relates the growth of the graded Helly numbers to the Helly number: graded Helly numbers must either grow fast or be stationary.
\begin{Prop}
    \label{prop:growing-helly-implies-helly-bounded}
    If $\h^{(t)}(\mathcal{F})< t$ for all $t>t_0$, then $\h({\mathcal{F}})\leq t_0$.
\end{Prop}
\begin{proof}
    Notice that we have $\h^{(t)}(\F) > \h^{(t-1)}(\F)$ if and only if $\h^{(t)}(\mathcal{F})=t$. 
    Thus, the assumption implies that the sequence $\h^{(t)}(\mathcal{F})$ is stationary starting from $t_0$. Since $\h^{(t)}(\mathcal{F})\xrightarrow[t\rightarrow\infty]{} \h(\mathcal{F})$, it follows that $\h(\mathcal{F})=\h^{(t_0)}(\mathcal{F})\leq t_0$.
\end{proof}
    In fact, the aforementioned condition characterizes all sequences that form graded Helly numbers. Every non-decreasing sequence $(u_t)\in\mathbb{N}^{\mathbb{N}}$ satisfying $u_t\leq t$ and $u_{t}>u_{t-1} \iff u_t=t$ corresponds to the graded Helly numbers of the family $\F=\left\{F_k^{(i)}=\{1,\dots,u_i\}\backslash \{k\} \;\mid\;i\in\mathbb{N},k\in\{1,\dots,u_i\}\right\}$ on the ground set $\N$. 

Note that establishing such a bridge between the graded Radon numbers and the Radon number seems challenging: the graded Radon numbers can grow slowly. We will see in Lemma~\ref{ex:growing-radon-number} a set system $\F$ where $\rad^{(t)}(\F)$ grows logarithmically with $t$. For now, we can only bound its growth from above.

\begin{Prop}
\label{prop:graded-radon-bounded}
 For any set system $\F$ and any $t \in \N$, we have $\rad^{(t)}(\mathcal{F})\leq t+1$.
\end{Prop}
\begin{proof}
Suppose there exists $\mathcal{F}' \subseteq \F$ of size $t$ of Radon number greater than $t+1$, that is with some set $S=\{p_1,\dots,p_{t+1}\}$ of $t+1$ points in $X$ with no $\F'$-Radon partition. 
If a partition $\mathcal{P}=(\mathcal{P}_0,\mathcal{P}_1)$ of $\{p_1,\dots,p_{t+1}\}$ is not a Radon partition, then there exists some $F_i\in\mathcal{F}'$ such that all members of $\mathcal{P}_0$ belong to $F_i$ and at least one member of $\mathcal{P}_1$ does not belong to $F_i$ (the $\F'$-hull of $\mathcal{P}_0$ does not contain any member of $\mathcal{P}_1$). Considering the 2-partitions $\mathcal{P}^{(j)}=(\{p_1,\dots,p_{t+1}\}\backslash \{p_j\},\{p_j\})$ for $j\in\{1,\dots,t\}$ yields $t$ sets $F_{i_j}\in\mathcal{F}'$ such that every $p_k$ belongs to $F_{i_j}$ except for $p_j$ who does not. This extra information implies that we get $t$ distinct indices $i_j$, that is $\{F_{i_1}, F_{i_2}, \ldots, F_{i_t}\} = \F'$. 
Hence, no $F \in\mathcal{F}'$ contains all points $p_1,\dots,p_t$: we get $\conv_{\mathcal{F}'}(\{p_1,\dots,p_t\})=X$ hence $(\{p_1,\dots,p_t\},\{p_{t+1}\})$ is a Radon partition, which contradicts the initial assumption.
\end{proof}

\paragraph{Toward a fractional Helly theorem.}
Holmsen \cite[Theorem~1.2]{Holmsen2020} proved that the colorful Helly number $\ch(\F)$ bounds the clique number from above when $\h(\F)\leq \ch(\F)$. 
Indeed, for every integers $k\leq m$, having $\ch(\F)\leq m$ and $\h(\F)\leq k$ means that some specific patterns (of arbitrary large size) are forbidden in the hypergraph of $k$-intersections. 
In particular, some patterns of fixed size $mk$ (called the complete $m$-tuples of missing edges, see the discussion in \cite[$\mathsection 3$]{Holmsen2021}) are forbidden in this hypergraph, which ensures $\cl_k(\F)\leq m$ from \cite[Theorem~1.2]{Holmsen2020}. 
However, complete $m$-tuples of missing edges are forbidden as soon as $\ch^{(mk)}(\F)\leq m$ and $\h^{(m)}(\F)\leq k$, implying $\cl_m(\F)\leq m$. More precisely: 
\medskip
\begin{Lemme}[{\cite[Theorem~1.2]{Holmsen2020}}]
    \label{lemma:graded-colorful-implies-clique}
    For $m\geq k>1$ and $\alpha\in(0,1)$, there exists $\beta(\alpha,k,m)\in(0,1)$ with the following property: For any family $\F$ such that $\ch^{(mk)}(\F)\leq m$ and $\h^{(m)}(\F)\leq k$ and for every finite subfamily $\F'\subset \F$, if a fraction $\alpha$ of the $m$-tuples of $\F'$ have nonempty intersection, then there exists a subfamily $\mathcal{G}\subset \F'$ of size $\beta(\alpha,k,m) |\F'|$ such that every $m$ sets of $\mathcal{G}$ intersect.
\end{Lemme}

The careful reader will note that, unlike in \cite[Theorem~1.1]{Holmsen2021}, our result does not yield a fractional Helly theorem. Recall that the fractional Helly theorem asserts that if a constant fraction of the $m$-tuples intersect, then a constant fraction of the elements of $\F'$ have a nonempty intersection. In our case, to ensure that all the elements of $\mathcal{G}$ intersect, it suffices to add the assumption $\h(\F)\leq m$. We then obtain a result analogous to \cite[Theorem~1.1]{Holmsen2021}.

\section{What about the homological shatter function ?}
In this section we are interested in families whose ground set is a topological space, for example $\R^d$. A parameter emerges, the homological complexity. \cite{hellyBetti} and \cite[Theorem~2.1]{patakova2022bounding} relate this parameter to the Helly and Radon numbers; we can deduce the same relations between their graded analogues by applying their results to each subfamily of size $t$. We get, for $\F$ a family of sets in $\R^d$:
\begin{align}
    \h^{(t)}(\F)< \rad^{(t)}(\mathcal{F})\leq r\left(\phi_{\mathcal{F}}^{(\lceil\frac{d}{2})\rceil}(t),d\right) && \text{ with } r(\cdot,\cdot) \text{ defined by \cite[Theorem~2.1]{patakova2022bounding}.} 
    \label{eq:bound-phi-bound-radon-patakova}
\end{align}
Notice that combining Inequality~\eqref{eq:bound-phi-bound-radon-patakova} with Inequality~\eqref{eq:colorful-helly} allows us to write
\begin{align}
    \ch^{(t)}(\F)\leq m\pth{r\pth{\phi_{\mathcal{F}}^{(\lceil\frac{d}{2})\rceil}(t),d}}.&&
    \label{eq:colorful-helly-bounded-by-phi}
\end{align}

\medskip

The conjecture of Kalai and Meshulam \cite[Conjectures 6 and 7]{Kalai2010CombinatorialEF}, \cite[Conjecture 1.9]{steppingUp} states that when the homological shatter function grows polynomially, the family $\F$ satisfies a fractional Helly theorem. Inequalities \eqref{eq:bound-phi-bound-radon-patakova} and \eqref{eq:colorful-helly-bounded-by-phi}, combined with Lemma~\ref{lemma:graded-colorful-implies-clique} are a step toward this conjecture: controlling the growth of the homological shatter functions allows to control the growth of the functions $\h^{(\cdot)}(\F),\rad^{(\cdot)}(\F), \ch^{(\cdot)}(\F)$, and hopefully bound from above some $\cl_k(\F)$. If we want to end up with a fractional Helly theorem rather than just a bounded $k$-clique number, we need the parameter $m$ 
to be greater than $\h(\F)$; in particular, we need $\h(\F)$ to be finite.

\paragraph{Bounding the Helly number.} We can deduce from Proposition~\ref{prop:growing-helly-implies-helly-bounded} and Inequation~\eqref{eq:bound-phi-bound-radon-patakova} that bounding the growth on the homological shatter function may allow to bound the Helly number of the family.  More precisely, we have 
\begin{align}
    \forall b'\geq b_0-1, \, \phi_{\F}^{\lceil(\frac{d}{2})\rceil}(r(b'+1,d))\leq b' \quad \Rightarrow \quad \h(\F)< r(b_0,d). \label{eq:psi-growth-bounded}
\end{align}
Indeed, the left hand side implies $\h^{(t)}(\F)<t$ for all $t\geq r(b_0,d)$, which implies the right hand side.

\paragraph{Bounding the $\h(\F)$-th clique number.} Let us fix a value $b_0$ such that Inequation~\eqref{eq:psi-growth-bounded} is satisfied. It remains to find the right value for $t=km_0$ to apply Lemma~\ref{lemma:graded-colorful-implies-clique}. As previously mentioned, we want to set $k\geq \h(\F)$, say $k=r(b_0,d)$. We now need a value $m_0$ such that $\ch^{(m_0r(b_0,d))}(\F)\leq m_0$. After applying Inequation~\eqref{eq:colorful-helly-bounded-by-phi}, it comes down to looking for a value $m_0$ such that 
\begin{align}
    m
    \pth{
        r\pth{
            \phi_{\mathcal{F}}^{(\lceil\frac{d}{2})\rceil}(m_0r(b_0,d))
            ,d }
        } 
    \leq m_0.
    \label{eq:value-m}
\end{align}

\paragraph{Definition of the function $\Psi_{d,b_0}$.} To summarize these two conditions, we introduce a non-stationary function $\Psi_{d,b_0}:\N\rightarrow \N$ that satisfies Inequation~\eqref{eq:psi-growth-bounded} and $\Psi_{d,b_0}(m(r(b_0,d))r(b_0,d))\leq b_0$. 
We can check that any function bounded from above by such a $\Psi_{d,b_0}$  satisfies Inequations \eqref{eq:psi-growth-bounded} and \eqref{eq:value-m} (since the functions $m$ and $r$ are non-decreasing). 

To make such a function explicit, we inverse $R_d:b'\mapsto r(b'+1,d)$ by defining 
$S:\mathbb{N}\rightarrow\mathbb{N}$ as $S(x)=\max\{b'\in\N\mid x\geq r(b'+1,d)\}$. 
Next, we define 
$\Psi_{d,b}: \mathbb{N}\rightarrow\mathbb{N}$ as follows:
\[\Psi_{d,b}(t)= \begin{cases}
    b-1 & \text{if }~t\leq r(b,d) \\
    b & \text{if }~ r(b,d)\leq t\leq m(r(b,d))r(b,d) \\
    S(t) & \text{if }~ t>m(r(b,d))r(b,d)
\end{cases}.\]
Note that the function $S$ is defined as the inverse of the rapidly growing function $R_d$, which seems to be (very roughly) bounded from above by a function in $\mathcal{E}^5$ in the Grzegorczyk hierarchy. Consequently, the best growth we can currently hope for $S$, and thus for $\Psi_{d,b}$, is only marginal: $S$ is somewhere in-between the inverse Ackermann ($\alpha$) function and the iterated logarithm ($log^*$).

\begin{Th}
    \label{th:main-theorem}
    For every integers $b,d\geq 0$, for every $\alpha\in (0,1)$, there exists $\beta = \beta(b,d,\alpha) >0$ and $n_0=n_0(b,d)$ with the following property: for every (possibly infinite) family $\mathcal{F}$ of sets in $\mathbb{R}^d$ with $\phi_\mathcal{F}^{(\lceil \frac{d}{2} \rceil)}$ bounded from above by $\Psi_{d,b}$, for any finite  
    subfamily $\mathcal{F}'\subset \mathcal{F}$ of size $|\F'|\geq n_0$
    with at least 
    $\alpha \binom{|\F'|}{d+1}$ 
    of its $d+1$-tuples 
    with non-empty intersection, 
    there exist some $\beta |\mathcal{F}'|$ members of $\mathcal{F}'$ that intersect.
\end{Th}

\begin{proof}
    Set $r_0=r(b_0,d)$ and $m_0=m(r)$. Let $n_0=n_0(b,d)$ be the "sufficiently large" threshold given in \cite[Theorem~1.2]{steppingUp}. 
    Let $\mathcal{F}'$ be a finite subfamily of $\mathcal{F}$ of size $|\F|\geq n_0$, such that at least 
    $\alpha \binom{|\F'|}{d+1}$ 
    of the $(d+1)$-tuples of $\mathcal{F}'$ intersect. Given that $\phi_{\mathcal{F}}^{(\lceil\frac{d}{2})\rceil}(m_0)\leq \Psi_{d,b_0}(m_0)\leq b_0$, \cite[Theorem~1.2]{steppingUp} guarantees that a proportion $\alpha'$ of the $m_0$-tuples of $\mathcal{F}'$ also intersect. 
    Since the function $\phi_{\F}^{(\lceil \frac{d}{2} \rceil)}$ satisfies Inequality~\eqref{eq:value-m}, it means that $\ch^{(m_0r_0)}(\F)\leq m_0$. We now meet the necessary conditions to apply Theorem~\ref{lemma:graded-colorful-implies-clique} by setting $k=r_0$: it yields a subfamily $\mathcal{G}$ of size $\beta(\alpha',r_0,m)|\mathcal{F}'|$, where all $r_0$ sets of $\mathcal{G}$ intersect. 
    In fact, it means that all the elements of $\mathcal{G}$ intersect, because the function $\phi_{\F}^{(\lceil \frac{d}{2} \rceil)}$ satisfies Inequality~\eqref{eq:psi-growth-bounded}, ensuring $\h(\F)\leq r_0$.
    \end{proof}

Note that even if the topology of the ground set (in our case, $\mathbb{R}^d$) sets the parameter $d$, the parameter $b$ may vary. For a given $\phi_{\mathcal{F}}^{(\lceil \frac{d}{2} \rceil)}$, we only need to find one value of $b$ such that $\phi_{\mathcal{F}}^{(\lceil \frac{d}{2} \rceil)}\leq \psi_{d,b}$ to apply the theorem. This condition becomes less restrictive when $b$ gets bigger, but the subfamily $\F'$ then has to be of larger size for the theorem to apply.

\section{Some examples}
\label{section:examples-betti-borne}
In this section we characterize homological shatter functions, and give an idea of how to characterize the growth of the graded Radon numbers.

\paragraph{Families with a given homological shatter function.}
The homological shatter function is non-decreasing. It turns out that all non-decreasing functions are the homological shatter function of some family.
\begin{Lemme}
    \label{lemma:construction-of-family-with-given-phi}
    For any non-decreasing function $f:\N\rightarrow \N$, for every $h\geq 0$ and $d\geq h+2$, there exists a family of sets $\F$ in $\R^d$ such that $\phi_{\F}^{(h)}=f$.
\end{Lemme}

\begin{proof}
    Consider countably many disjoint filled boxes 
    in $\R^d$. For $i\geq d-h$, we define a family $\F_i\coloneqq\{F_1^{(i)},\dots,F_{i}^{(i)}\}$ as follows (for $i<d-h$, we can adapt the construction): 
    place $f(i)$ disjoint $(d-h-1)$-dimensional polytopes with $i$ facets inside the $i$-th rectangle, and label the facets of each polytope from $1$ to $i$. 
    Define the set $F_k^{(i)}$ as the complement of all the facets labeled $k$ within the filled rectangle. The intersection between all members of $\mathcal{F}_i$ is the complement of the boundaries of these polytopes inside the filled rectangle. Using Alexander's duality, we know that the $h$-th Betti number of this intersection is $\tilde{\beta}_{h}(\bigcap_{F\in\mathcal{F}_i} F,\Z_2)=f(i)$. 
Moreover, intersections between fewer sets, say $\mathcal{F}_i'$, satisfy $\tilde{\beta}_{h'}(\bigcap_{F\in\mathcal{F}_i'} F,\Z_2)=0$ for $h'<h$. 
Thus, the family $\mathcal{F}\coloneqq\bigcup_{i\in\N} \mathcal{F}_i$ has its homological shatter function $\phi^{(h)}_{\mathcal{F}}$ equal to the function $f$.
\end{proof}

For every $d>3$, the integer $h=\lceil\frac{d}{2}\rceil$ satisfies $d\geq h+2$: the family obtained from Lemma~\ref{lemma:construction-of-family-with-given-phi} is not suited to apply \cite[Corollary 1.3]{steppingUp} but suited to apply Theorem~\ref{th:main-theorem}. Naturally, we do not need to apply Theorem~\ref{th:main-theorem} to get a fractional Helly number: the nerve of our family is a disjoint union of cliques, thus the usual fractional Helly theorem can be applied. On another note, notice that we do not control at all the ($h'$th) homological shatter functions for $h'>h$. 

\paragraph{Growth of the Radon numbers.} 

Similarly to the graded Helly numbers, it would be useful to characterize the indices $t$ such that $\rad^{(t)}(\F)>\rad^{(t-1)}(\F)$. A first step would be to understand the pairs $(\F,S)$ where $\F$ is a family of sets $\F$ and $S$ is a point set of the ground set such that $S$ is \textit{minimally non-partitionable for $\F$}, that is:
\begin{itemize}
    \item[--] $S$ admits no $ \F$-Radon partition, but
    \item[--] For every $F\in\F$, $S$ admits an $(\F\backslash \{F\})$-Radon partition.
\end{itemize}

\begin{Lemme}
    For every $k\geq 2$, there exist a family $\F$ of size $2^{k-1}-1$ and a set $S$ of size $k$ such that $S$ is minimally non-partitionable for $\F$.
\end{Lemme}
\begin{proof}
    \label{ex:growing-radon-number}
    Set $X=\{0,1\}^{2^{k-1}-1}$ the set of binary words of length $2^{k-1}-1$ and consider the family $\F = \{F_i^{\delta}\mid i\in\mathbb{N},\delta\in\{0,1\}\}$, where $F_i^{\delta}$ is the set of words whose $i$-th letter is $\delta$. 

  An $\mathcal{F}$-convex set is now a cartesian product of occurrences of $\{0\},\{1\}$ or $\{0,1\}$. (Note that $\mathcal{F}$ has Helly number $2$.) If $S$ is a set of words, consider the indices where all the words share the same letter: the $\mathcal{F}$-hull of $S$ is the set of words that match these letters at the corresponding indices. Notice that two subsets ${\cal A}, {\cal B} \subset X$ have disjoint $\F$-hulls if and only if there exists an index $i$ such that all words of ${\cal A}$ have their $i$-th letter equal to $\delta$ and all words of ${\cal B}$ have their $i$-th letter equal to $1-\delta$.
    When we consider the $\F'$-hull for a subfamily $\mathcal{F}'\subseteq \F$, the same holds considering only the indices $i\in\mathbb{N}$ such that $\{F_i^0,F_i^1\} \subset\mathcal{F}'$.

    Denote all the $2$-partitions of $\{1,\dots,k\}$ by $\mathcal{P}^{(1)},\dots,\mathcal{P}^{(2^{k-1}-1)}$. For $j\in\{1,\dots,k\}$, we pick $p_j$ to be the word whose $i$-th letter is $0$ or $1$ depending on whether $j\in\mathcal{P}^{(i)}_0$ or $j\in\mathcal{P}^{(i)}_1$. We consider the set of words $S = \{p_1,\dots,p_k\}$. This family $\F$ and this set $S$ are such that:
    \begin{itemize}
        \item[--]  $S$ has no $\F$-Radon partition. Pick any 2-partition of $S$, say $(\{p_i\in S\mid i\in \mathcal{P}^{(\ell)}_0\},\{p_i\in S\mid i\in \mathcal{P}^{(\ell)}_1\})$. For every $j\in \mathcal{P}^{(\ell)}_1$, the $\ell$-th letter of the word $p_j$ is $1$, and for every $j\in \mathcal{P}^{(\ell)}_0$, the $\ell$-th letter of the word $p_j$ is $0$: the $\F$-convex hulls do not intersect, thus the partition $\mathcal{P}^{(s)}$ does not form an $\F$-Radon partition.
\item[--] For every $F_{\ell}^{\delta}\in\F$, the partition $(\{p_i\in S\mid i\in \mathcal{P}^{(\ell)}_0\},\{p_i\in S\mid i\in \mathcal{P}^{(\ell)}_1\})$ is an $(\F\backslash\{F_{\ell}^{\delta}\})$-Radon partition for $S$.\qedhere
\end{itemize}
\end{proof}
By considering the set of infinite binary words and the corresponding family $\{F_i^{\delta}\}$, we get a family whose graded Radon numbers grow logarithmically (the exact graded Radon numbers can be computed following a reasoning similar to Proposition~\ref{prop:graded-radon-bounded}). In particular, its Radon number cannot be bounded.

\bibliographystyle{plain}

\end{document}